\documentclass[reqno, 12 pt]{amsart}

\usepackage{amsthm}
\usepackage{amssymb}
\usepackage{latexsym}
\usepackage{url}
\usepackage{amsmath}
\usepackage{verbatim}
\usepackage{graphicx}
\usepackage{epsf}
\usepackage{enumerate}
\usepackage{geometry}
\usepackage[width=6.4cm]{caption}
\usepackage{hyperref}
\hypersetup{breaklinks=true,
pagecolor=white,
colorlinks=true,
citecolor=black,%
filecolor=black,%
linkcolor=black,%
urlcolor=black
}

\newtheorem{thm}{Theorem}[section]

\newtheorem{lem}[thm]{Lemma}
\newtheorem{prop}[thm]{Proposition}

\newtheorem{defin}[thm]{Definition}

\theoremstyle{definition}
\newtheorem{expl}[thm]{Example}
\newtheorem{remark}[thm]{Remark}

\newcommand{\dom}{\textnormal{dom}}
\newcommand{\Dom}{\textnormal{Dom}}

\newcommand{\R}{\mathbb{R}}

\newcommand{\N}{\mathbb{N}}

\newcommand{\wt}{\widetilde}

\newcommand{\vol}{\textnormal{vol}}

\title[On the existence of a neutral region]{On the existence of a neutral region} 
\author{Daniel Reem} 
\address{\noindent IMPA - Instituto Nacional de Matem\'atica Pura e Aplicada, Estrada Dona Castorina 110, Jardim Bot\^anico, CEP 22460-320, Rio de Janeiro, RJ,  Brazil.} 
\email{\noindent dream@impa.br } 
\begin{document}
\maketitle 
\begin{abstract}
Consider a given space, e.g., the Euclidean plane, and its 
decomposition into Voronoi regions induced by given sites. 
It seems intuitively clear  that each point in the space 
belongs to at least one of the  regions, i.e., no neutral 
region can exist. As simple counterexamples show this is not 
true in general, but we present a simple necessary 
and sufficient condition ensuring the non-existence of  a neutral 
region. We discuss a similar phenomenon concerning recent 
variations of Voronoi diagrams called zone diagrams, 
double zone diagrams, and (double) territory diagrams. 
These objects are defined in a somewhat implicit way and 
they also induce a decomposition of the space into regions. 
In several works it was claimed without providing a proof that some of these objects induce a 
decomposition in which a neutral region must exist.  We show that this assertion is true 
in a wide class of cases but not in general. We also discuss other 
properties related to the neutral region, among 
them a one related to the concentration of measure phenomenon.  
\end{abstract}

\section{\bf Introduction}\label{sec:Introduction}
Consider a given space, e.g., the Euclidean plane, and its decomposition into 
 Voronoi regions (Voronoi cells) induced by given sites. It seems 
intuitively clear that the regions form a subdivision, i.e., 
each point in the space belongs to at least one of the 
 regions. As a matter of fact, this is claimed  in various places, e.g., 
in \cite[pp. 345-6]{Aurenhammer},\cite[p. 513]{Fortune2004}, \cite[p. 47]{Mulmuley1994}. 
In these places it is assumed (explicitly or implicitly) 
that the number of sites is finite, 
an assumption which obviously implies the subdivision property. 
However, the assumption of finitely many sites is not always satisfied, e.g., in the case of Voronoi  diagrams 
in the context of the lattices such as in the geometry of numbers or crystallography 
or stochastic geometry (see Example \ref{ex:finite}). 
It turns out that in general a neutral cell may indeed exist, and hence 
one may ask whether it is possible to formulate a simple necessary and sufficient 
condition ensuring that no such a region exists. Such a condition is formulated in Section \ref{sec:NeutralVoronoi}, 
and is illustrated using various examples. To the best of our knowledge, 
the possibility of the existence of a neutral Voronoi cell was published  
 only in \cite{ReemISVD09,ReemGeometricStabilityArxiv}, but no systematic investigation 
 of this issue was carried out there. 

 A related phenomena, now viewed from the reverse direction, appears in connection with 
 recent variations of Voronoi diagrams called zone diagrams \cite{AMTn,KMT,KopeckaReemReich, ReemReichZone}. 
 As in the case of Voronoi diagrams, these geometric objects induce  
 a decomposition of the given space into regions, but in contrast with the Voronoi diagram, 
 in which the region $R_k$ associated with the site $P_k$ is the set of all points in the space 
 whose distance to $P_k$ is not greater than their distance to the other sites $P_j,j\neq k$, 
 in the case of a zone diagram the region $R_k$ is the set of all the points in the space 
 whose distance to $P_k$ is not greater to their distance to the other regions $R_j, j\neq k$. 

This somewhat implicit definition implies, after some thinking, that  a zone diagram 
 is a solution to a certain fixed point equation. Although its existence is not obvious 
 in advance, it seems clear that if a zone diagram does exist, it induces a decomposition 
 of the space into the regions (zones) $R_k$, and an additional region: the neutral one. 
 See Figure \ref{fig:NeutralZD_l2_0001_82Site_1inSite}. 
 This actually was claimed explicitly in several places \cite{AMT2,AMTn,DeBiasiKalantaris2011}, 
but this claim has not been proved.

\begin{figure}[t]
\begin{minipage}[t]{0.45\textwidth}
\begin{center}
{\includegraphics[scale=0.63]{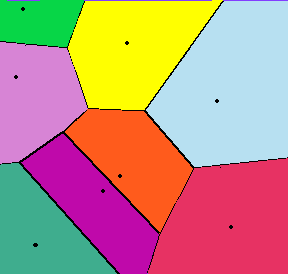}}
\end{center}
 \caption{Voronoi diagram of 8 point sites in a square in $(\R^2,\ell_2)$. No neutral region exists.}
\label{fig:NeutralVD_l2_0001_8Site_1inSite}
\end{minipage}
\hfill
\begin{minipage}[t]{0.45\textwidth}
\begin{center}
{\includegraphics[scale=0.63]{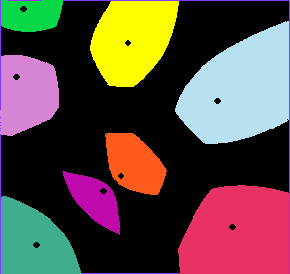}}
\end{center}
 \caption{The zone diagram [and hence a double zone diagram and a (double) territory diagram], 
 of the same sites given in Figure \ref{fig:NeutralVD_l2_0001_8Site_1inSite}. 
 The (black) neutral region is clearly seen.} 
\label{fig:NeutralZD_l2_0001_82Site_1inSite}
\end{minipage}
\end{figure}

As a matter of fact, the very first works discussing the concept of a zone diagram used the terminology 
``a Voronoi diagram with neutral zones''  \cite[p. 25]{AsanoTokuyama} and 
``Voronoi diagram with neutral zone'' (pages 336-8 and 343 of the 2006 conference version of  \cite{AMT2}, 
 the bottom of \cite[p. 1182]{AMTn}) for describing this concept. 
 In Section~\ref{sec:NeutralZone} we prove that the above claim about the existence of a neutral region 
 holds in a wide class of spaces (geodesic metric spaces) but not in general. 
We discuss similar phenomena occurring with variations of zone diagrams called double zone diagrams 
\cite{ReemReichZone}, territory diagrams \cite{DeBiasiKalantaris2011} 
(called subzone diagrams in the conference version of \cite{DeBiasiKalantaris2011}), and 
double territory diagrams which are introduced here (we also generalize the definition of territory diagrams 
from the setting of the Euclidean plane with point sites).  
Again, the existence of a neutral zone in the case of territory diagrams was claimed without any proof. 

Geodesic metric spaces satisfying a certain compactness assumption were discussed in the related setting  
of $k$-sectors and $k$-gradations \cite{ImaiKawamuraMatousekReemTokuyama2010CGTA}. 
Here however no compactness assumption is made. As shown in Section \ref{sec:Interpretation}, 
at least in this setting the neutral region can be used to justify the interpretation 
of zone diagram as a certain equilibrium between mutually hostile opponents. This interpretation 
 was mentioned without full justification in \cite{AMTn,ReemReichZone}. 
 In Section \ref{sec:ConcentrationOfMeasure} we show that not only 
 the neutral region is nonempty but it actually occupies a volume much larger  
than the volume of the ``interior regions''. 
 Here one considers double zone diagrams of separated point sites in a finite dimensional 
 Euclidean space. This phenomenon   is related  to (but definitely distinguished from) the phenomenon 
known as ``concentration of measure'' \cite[pp. 165-166]{Gruber2007},\cite[pp. 329-341]{Matousek2002}. 
The paper ends in Section ~\ref{sec:ConcludingRemarks} 
with a few remarks about possible lines of further investigation.

\section{\bf Preliminaries}\label{sec:Preliminaries} 
In this section we present our notation and basic definitions, as well as additional details 
regarding the basic notions. Throughout the text we will make use of tuples, the components of which are sets (which are subsets of a given set $X$). Every operation or relation between such tuples, or on a single tuple, is done component-wise. Hence, for example, if $K\neq \emptyset$ is a set of indices, and if $R=(R_k)_{k\in K}$ and $S=(S_k)_{k\in K}$ are two tuples  of subsets of $X$, then 
$R\subseteq S$ means $R_k\subseteq S_k$ for each $k\in K$. When $R$ is a tuple, the  notation $(R)_k$ is the $k$-th component 
of $R$, i.e, $(R)_k=R_k$. 

\begin{defin}\label{def:dom}
Given two nonempty subsets $P,A\subseteq X$, the dominance region 
$\dom(P,A)$ of $P$ with respect to $A$ is the set of all $x\in X$
whose distance to $P$ is not greater than their distance to $A$, i.e.,
\begin{equation}\label{eq:dom}
\dom(P,A)=\{x\in X: d(x,P)\leq d(x,A)\}.
\end{equation}
Here 
\begin{equation}\label{eq:dxA}
d(x,A)=\inf\{d(x,a): a\in A\} 
\end{equation}
and in general, for any  subsets $A_1,A_2$ we denote 
\begin{equation*}
d(A_1,A_2)=\inf\{d(a_1,a_2): a_1\in A_1,\,a_2\in A_2\}.
\end{equation*}
with the agreement that $d(A_1,A_2)=\infty$ if $A_1=\emptyset$ or $A_1=\emptyset$. 
\end{defin}

\begin{defin}\label{def:Voronoi}
Let $K$ be a set of at least 2 elements (indices), possibly
infinite. Given a  tuple $(P_k)_{k\in K}$ of nonempty subsets
$P_k\subseteq X$, called the generators or the sites, the Voronoi diagram  induced by this tuple is the tuple $(R_k)_{k\in K}$ of nonempty subsets
$R_k\subseteq X$, such that for all $k\in K$,
\begin{equation}\label{eq:VoronoiCell}
R_k=\dom(P_k,{\underset{j\neq k}\bigcup P_j})=\{x\in X: d(x,P_k)\leq d(x,P_j)\,\,\forall j\neq k,\,j\in K \}.
\end{equation}
 In other words, each $R_k$, called a Voronoi cell or a Voronoi region, is the set of
all $x\in X$ whose distance to $P_k$ is not greater than
its distance to any other site $P_j$, $j\neq k$. The set $X\backslash (\bigcup_{j\in K}R_j)$ is called 
the neutral region. 
\end{defin}

\begin{remark}\label{rem:GeneralDistance}
Voronoi diagrams can be defined in a more general context than metric spaces, and actually 
we will use such a setting later (Theorem \ref{thm:NoNeutral}). As in Definition \ref{def:Voronoi}, 
one starts with a nonempty set $X$ and a tuple $(P_k)_{k\in K}$ of nonempty subsets of 
$X$. However, now the distance function is $d:X^2\to [-\infty,\infty]$ and it is not 
limited to satisfy the axioms of a metric (e.g., the triangle inequality, being nonnegative 
and symmetric, etc.). The dominance region is defined as in \eqref{def:dom}, the distance 
$d(x,A)$ is defined as in \eqref{eq:dxA}, and  Voronoi cells are defined exactly 
as in \eqref{eq:VoronoiCell}. Voronoi diagrams based on Bregman distance \cite{BoissonnatNielsenNock2010}, 
convex distance functions \cite{ChewDrysdale}, and other examples are all particular 
cases of this setting. Such a setting, which seems to not has been discussed before, 
even generalizes the setting of $m$-spaces \cite{ReemReichZone} 
in which the only restriction on $d$ is that $d(x,x)\leq d(x,y)$ for any $x$ and $y$ (the 
previously mentioned cases are actually particular cases of $m$-space since 
in them $0=d(x,x)\leq d(x,y)$). 
\end{remark}

\begin{defin}\label{def:zone}
Let $K$ be a set of at least 2 elements (indices), possibly
infinite. Given a  tuple $(P_k)_{k\in K}$ of nonempty subsets
$P_k\subseteq X$, a zone diagram with respect to
that tuple is a tuple $R=(R_k)_{k\in K}$ of nonempty subsets
$R_k\subseteq X$ satisfying 
\begin{equation*}
R_k=\dom(P_k,\textstyle{\underset{j\neq k}\bigcup R_j})\quad \forall k\in K.
\end{equation*}
In other words, if we define $X_k=\{C: P_k\subseteq C\subseteq X\}$, then a zone diagram
is a fixed point of the mapping $\Dom:\underset{{k\in K}}\prod
X_k\to \underset{{k\in K}}\prod
X_k$, defined by 
\begin{equation}\label{eq:TZoneDef}
\Dom(R)=(\dom(P_k,\textstyle{\underset{j\neq k}\bigcup R_j}))_{k\in K}.
\end{equation}
A tuple $R=(R_k)_{k\in K}$ is called a double zone diagram if it is the fixed point of the second iteration $\Dom\circ \Dom$, i.e., $R=\Dom^2(R)$. A tuple $R=(R_k)_{k\in K}$ is called a territory diagram if $R\subseteq \Dom(R)$ and 
it is called a double territory diagram if $R\subseteq \Dom^2(R)$. 
\end{defin}

\begin{remark}\label{rem:Z-DZ-T-DT}
Some of the concepts mentioned in Definition \ref{def:zone} are related. 
Any zone diagram is obviously a territory diagram. It is also 
a double zone diagram as can be seen by applying $\Dom$ on $R=\Dom(R)$. 
Any double zone diagram is obviously a double territory diagram. 
A double territory diagram is not necessarily a territory diagram: 
take $X=\{-1,0,1\}\subset \R$, $(P_1,P_2)=(\{-1\},\{1\})$, 
$R=(\{-1,0\},\{0,1\})$; then $\Dom(R)=(\{-1\},\{1\})$ and 
$R \subsetneqq\Dom(R)$. A territory 
diagram is not necessarily a double territory diagram: take 
$X=[-1,1]\subset \R$,  $(P_1,P_2)=(\{-1\},\{1\})$, 
$R=([-1,0],\{1\})$; then $\Dom(R)=([-1,0],[0.5,1])$, 
$\Dom^2(R)=([-1,-0.25],[0.5,1])$, and hence $R\subsetneqq \Dom^2(R)$. 
\end{remark}

\begin{remark}\label{rem:TerritoryVoronoi}
The components of any territory and double territory 
diagrams are contained in the Voronoi cells of their sites. Indeed, 
the Voronoi cells corresponding to the tuple $P=(P_k)_{k\in K}$ 
of sites is nothing but $\Dom(P)$. By the definition $\Dom$ and the 
space $\prod_{k\in K} X_k$ of tuples we have 
$P\subseteq R$ and $P\subseteq \Dom(R)$ for any tuple $R$ in this space. Thus the 
anti monotonicity of $\Dom$ (see Lemma \ref{lem:ImportedLem}\eqref{item:Monotone}) 
implies that 
$\Dom(R)\subseteq\Dom(P)$ and $\Dom^2(R)\subseteq \Dom(P)$ and 
the assertion follows by taking $R$ to be a territory or a double territory diagram. 
\end{remark}

\begin{remark}\label{rem:TeritoryExamples}
Examples (illustrations) of 2-dimensional zone diagrams in various settings 
can be found in \cite{AMTn,ImaiKawamuraMatousekReemTokuyama2010CGTA,
KopeckaReemReich,ReemZoneCompute, ReemReichZone}. 
Examples of double zone diagrams which are not zone diagrams can 
be found in \cite{ReemZoneCompute,ReemReichZone}. 
 Examples (including illustrations) of territory diagrams which are not zone 
diagrams can be found in \cite{DeBiasiKalantaris2011}. Additional illustrations 
can be found in Figures \ref{fig:NeutralZD_l2_0001_82Site_1inSite}, 
\ref{fig:ZD-4iterations-2sites-2inSite-0002-New}-\ref{fig:DZ-Max-0002-2sites2inSite-l1}, 
and \ref{fig:ZD-Flowers-4Iteration-17Sites-0005}. 

Existence (and sometimes uniqueness) proofs of zone diagrams in certain settings 
can be found in \cite{AMTn,KMT,KopeckaReemReich,ReemReichZone}. For 
our purposes we only need to know that double zone diagrams 
always exist \cite{ReemReichZone} and the same is true for territory 
diagrams and double territory diagrams. As a matter of fact, it is quite easy to construct explicit examples of territory 
and double territory diagrams: we can simply 
start with $P=(P_k)_{k\in K}$ and iterate it using $\Dom$. As explained in 
Remark ~\ref{rem:TerritoryVoronoi}, for each tuple $R$ one has $P\subseteq \Dom(R)$ and $P\subseteq \Dom^2(R)$. 
Now, since $\Dom$ is antimonotone and since $\Dom^2$ is monotone the inequality  
\begin{equation*}
P\subseteq \Dom^2(P)\subseteq \Dom^4(P)\subseteq\ldots\subseteq\ldots
\subseteq \Dom^3(P)\subseteq\Dom(P)
\end{equation*}
follows. Hence any even power is a territory and double territory diagram [which 
is usually not a (double) zone diagram].
\end{remark}

We finish this section with the definition of geodesic metric spaces. 
\begin{defin}\label{def:GeodesicMetric}
Let $(X,d)$ be a metric space. Let $x,y\in S\subseteq X$.  The subset $S$ is
called a metric segment between $x$ and $y$  if there exists
 an isometry $\gamma$ (i.e., $\gamma$ preserves distances) from the real line segment $[0,d(y,x)]$ onto $S$ such that $\gamma(0)=x$ and  $\gamma(d(y,x))=y$. We denote $S=[x,y]_{\gamma}$, or simply $S=[x,y]$. If between all points $x,y\in X$ there exists a metric segment, then $(X,d)$ is called  a  geodesic metric space.
\end{defin}
Simple and familiar examples of geodesic metric spaces are: the Euclidean plane, any normed space, 
any convex subset of a normed space, spheres, complete Riemannian manifolds \cite[pp. 25-28]{Jost}, 
and hyperbolic spaces \cite[pp. 538-9]{ReichShafrir}.

\section{\bf A neutral Voronoi region}\label{sec:NeutralVoronoi}
As mentioned in the introduction, although it might seem somewhat surprising, there are simple 
examples showing the existence of a neutral Voronoi. See Figure~ \ref{fig:NeutralRegionPointSites}. 
A quick glance at this figure shows that there are infinitely many sites. On the other hand, it can be easily verified 
that a sufficient condition for the non-existence of a neutral region is having finitely many sites. 
 But is it a necessary condition? The answer is no, as shown in Example \ref{ex:finite}. A more careful look at Figure  \ref{fig:NeutralRegionPointSites}  shows that the set obtained from taking the union of the sites has an 
 accumulation point, while  in the cases mentioned in Example \ref{ex:finite} no such accumulation point exists. 
 Hence it is natural to guess that a neutral region does 
 not exist if and only if no accumulation point exists. This is indeed true whenever the dimension if 
 finite or the space is compact (see Proposition \ref{prop:Accumulation}), 
 but Example \ref{ex:NoCluster} presents a simple infinite dimensional counterexample. 

As a result, if one is interested in a general necessary and  sufficient condition, then a different property 
should be detected. It turns out that this property is nothing but the existence of 
a nearest site. This property holds in any metric space, and, interestingly, actually even when 
most of the properties of the distance function (e.g., the triangle inequality, symmetry, non-negativeness) 
are removed (see Remark \ref{rem:GeneralDistance} for a discussion on Voronoi diagrams obtained 
from such distance functions). Despite this general setting, the proof of the corresponding theorem 
below is very simple. 
 
\begin{thm}\label{thm:NoNeutral}
Let $X$ be a nonempty set and let $d:X^2\to[-\infty,\infty]$. 
Given a tuple of nonempty subsets $(P_k)_{k\in K}$ contained in $X$, 
there exists no neutral Voronoi region if and only if for each $x\in X$ 
there exists a nearest site, namely,  there exists $j\in K$ such that 
\begin{equation}\label{eq:NN}
\inf\{d(x,P_k): k\in K\}=d(x,P_{j}). 
\end{equation}
Equivalently, $d(x,\cup_{k\in K}P_k)$ is index attained, namely there exists $j\in K$ 
such that 
\begin{equation}\label{eq:Index}
d(x,\bigcup_{k\in K}P_k)=d(x,P_{j}). 
\end{equation}
\end{thm} 
\begin{proof}
If for each $x\in X$ there exists $j\in K$ such that \eqref{eq:NN} holds,  
then  $d(x,P_j)\leq d(x,P_k)$ for any $k\in K$. Thus every $x\in X$ is 
in the Voronoi cell of some site $P_j$ and hence 
the neutral region is empty. On the other hand, suppose 
that there exists no neutral region, that is, any $x\in X$ belongs to the cell of $P_j$  
for some $j\in K$. This means that $d(x,P_j)\leq d(x,P_k)$ for any $k\in K$. Thus 
$d(x,P_j)\leq \inf\{d(x,P_k): k\in K\}$. But obviously $\inf\{d(x,P_k): k\in K\}\leq d(x,P_j)$ 
since $j\in K$. This implies equality and proves the first part of the assertion. 

The second part, namely \eqref{eq:Index}, is a simple consequence of the fact that 
\begin{equation*}
\alpha:=d(x,\bigcup_{k\in K}P_k)=\inf\{d(x,P_k): k\in K\}=:\beta. 
\end{equation*}
Indeed,  $\alpha\leq d(x,P_k)$ for all $k\in K$ by the definition of $\alpha$, 
so $\alpha\leq \beta$.  If $\alpha<\beta$, then there is 
$y\in \cup_{k\in K}P_k$ such that $d(x,y)<\beta$. Since $y\in P_{k}$ 
for some $k\in K$ we have $d(x,P_{k})\leq d(x,y)<\beta$, a contradiction 
with the definition of $\beta$.
\end{proof}
The following proposition gives additional  sufficient and necessary conditions for the non-existence 
of a neutral region in a more familiar setting. The property described in 
Part \eqref{item:FinitelyCompact} is sometimes called finitely compactness \cite{KapelusznyKuczumowReich} 
and it holds, e.g., when the space is compact or finite dimensional. 

\begin{prop}\label{prop:Accumulation}
Let $(X,d)$ be a metric space. Let $(P_k)_{k\in K}$ be a tuple of nonempty subsets in $X$. 
\begin{enumerate}[(a)]
\item\label{item:External} If there exists no neutral region in $X$ and the sites are closed sets, 
then $\bigcup_{j\in K}P_j$ has no external accumulation point (an accumulation point $y\notin \bigcup_{j\in K}P_j$). 
\item\label{item:FinitelyCompact} Suppose that $(X,d)$ has the property that any bounded infinite subset 
has an accumulation point. If $\bigcup_{j\in K}P_j$ has no accumulation points, then there does not exist a 
neutral region in $X$. 
\end{enumerate}
\end{prop}

\begin{proof}
We first prove \eqref{item:External}. Suppose by way of negation that some $x\in X$ is an external accumulation point of 
$\cup_{j\in K}P_j$. We claim that $x$ does not have any nearest neighbor. 
Indeed, let $k\in K$. Then $r=d(x,P_k)>0$, otherwise  $x\in P_k\subseteq \bigcup_{j\in K}P_j$, 
a contradiction. Since $x$ is an accumulation point of $\cup_{j\in K}P_j$, the open ball $B(x,r)$ contains a 
point $y\in \bigcup_{j\in K}P_j$. By the definition of $r$ we have $y\in P_{i}$ for 
some $i\neq k$, $i\in K$. Hence $d(x,P_{i})\leq d(x,y)<d(x,P_k)$. Since $k\in K$ was arbitrary 
this shows that no site $P_k$ can be a nearest site of $x$. 
By Theorem \ref{thm:NoNeutral} it follows that $x$ is in the neutral region, a contradiction. 

Now consider Part \eqref{item:FinitelyCompact} and suppose by way of contradiction that 
the neutral region is nonempty. Let $x$ be some point 
in the neutral region. Let $k_1\in K$. Theorem \ref{thm:NoNeutral} implies that 
$P_{k_1}$ is not the nearest site of $x$. Therefore $d(x,P_{k_2})<d(x,P_{k_1})$ for 
some $k_2\neq k_1$, $k_2\in K$. In particular $r_1=d(x,P_{k_1})>0$ and 
there exists a point $x_2\in P_{k_2}$ in the open ball $B(x,r_1)$. As before,  
 Theorem \ref{thm:NoNeutral} implies that $P_{k_2}$ is not the nearest site of $x$. 
  Therefore $d(x,P_{k_3})<d(x,P_{k_2})$ for 
some $k_3\neq k_2$, $k_3\in K$. We continue in this way and construct an infinite 
sequence $(P_{k_n})_{n=1}^{\infty}$ of different sites having the property 
that $d(x,P_{k_{n+1}})<d(x,P_{k_{n}})$ for any $n\in \N$. Hence there exists 
a sequence of points $(x_n)_{n=2}^{\infty}$ satisfying 
$d(x,P_{k_n})\leq d(x,x_n)<d(x,P_{k_{n-1}})\leq d(x,x_{n-1})$ and $x_n\in P_{k_n}$ 
for any $n\in \N$. In particular no two points from this sequence coincide. 
This set of points is an infinite set contained in the bounded ball $B(x,r_1)$. 
Hence it has an accumulation point, which is obviously an accumulation point 
of $\cup_{j\in K}P_j$, a contradiction.
\end{proof}

\begin{expl}\label{ex:finite}
The nearest site condition mentioned in Theorem \ref{thm:NoNeutral} obviously 
holds when $K$ is finite. A simple verification shows that the condition also holds 
when for each $x\in X$ there exists a ball centered at $x$ which 
intersects finitely many sites from $\cup_{j\in K}P_j$, since in this case 
the nearest site is one of the finitely many sites intersected by the ball 
(the intersection may include infinitely many points, but they belong to finitely many sites). 
This happens, e.g., when the sites form a lattice, as in the case of the geometry of numbers
in $\R^m$ \cite{GruberLek},  crystallography \cite{AshcroftMermin} (under the names  
``the Brillouin zone'' or ``the Wigner-Seitz cell''), 
 coding \cite[pp. 66-69,  451-477]{ConwaySloane}, or a somewhat random (infinite) distribution, such 
 as  Poisson Voronoi diagrams \cite[pp. 39, 291-410]{OBSC}. Another example:  $P_k=\R\times \{k\}$, $k\in \N$.
\end{expl}

\begin{expl}\label{ex:H}
A simple example when the nearest site condition fails: 
 $(X,d)$ is the Euclidean plane,  $P_k=\{(0,a_k)\}$ for all $k\in\N$ 
 (or, alternatively, $P_k=\R\times\{a_k\}$  for all $k\in\N$) where $a_k>0$  and
  $\lim_{k\to\infty}a_k=0$. The lower 
halfspace $H=\{(x_1,x_2): x_2\leq 0\}$ is the neutral region. See Figure~ \ref{fig:NeutralRegionPointSites}. 
A variation of this example was mentioned in \cite{ReemISVD09}. 
\end{expl}
\begin{figure}
\begin{minipage}[t]{0.42\textwidth}
\begin{center}
{\includegraphics[scale=0.5]{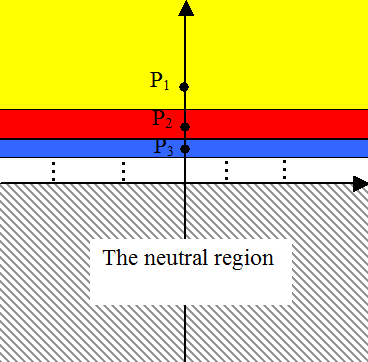}}
\end{center}
 \caption{A neutral Voronoi region induced by infinitely many point sites converging to 
 the origin  in the Euclidean plane (Example \ref{ex:H}).} 
\label{fig:NeutralRegionPointSites}
\end{minipage}
\end{figure}

\begin{expl}\label{ex:cluster}
An illustration of Proposition \ref{prop:Accumulation}\eqref{item:External} was actually given 
in Example ~\ref{ex:H}:  the point $(0,0)$ is the unique (external) accumulation point. The neutral 
region is however a much larger set than the set of accumulation points of $\cup_{k\in K}P_k$. 
 As another illustration of Proposition ~\ref{prop:Accumulation}\eqref{item:External}, take $S$ to be a 
 dense set in $X$ which is not $X$, e.g., the set of all points in the 
plane with rational coordinates, and let $K=S$. For each $k\in K$ define $P_k=\{k\}$. Then the neutral 
region is the complement of $S$. 
\end{expl}

\begin{expl}\label{ex:NoCluster}
This example shows that a neutral region  may exist  if the union $\cup_{j\in K}P_j$  does not 
have any accumulation point but the space is not finitely compact. Let $(X,d)$ to be the infinite dimensional space $\ell_2$ of all sequences $(x_n)_{n=1}^{\infty}$ 
of real numbers satisfying $\sum_{n=1}^{\infty}|x_n|^2<\infty$. Let $K=\N$ and let $e_k$ be the $k$-th basis element, 
i.e., the sequence whose $k$-th component is 1 and the other components are 0. 
Let  $P_k=\{((k+1)/k)e_k\}, k\in K$ be the sites. 
Then the point $x=0$ does not have any nearest site since $d(x,P_{k+1})<d(x,P_k)$ for any $k\in K$. 
Hence it in the neutral region. However, $\bigcup_{j\in K}P_j$  does not have any accumulation point 
since $d(P_k,P_j)\geq \sqrt{2}$ for any $k,j\in K$, $k\neq j$.  
\end{expl}

\section{\bf A neutral (double, territory) zone}\label{sec:NeutralZone}
In this section we discuss the existence of a neutral region (zone) in the context 
of zone diagrams, double zone diagrams, and (double) territory diagrams. 
We need the following lemma whose proof can be found e.g., in \cite[Lemma 5.4]{ReemReichZone} 
(Part ~\eqref{item:Monotone}) and \cite[Lemma 6.3, Lemma 6.8, Remark 6.9]{ReemZoneCompute} (Parts  \eqref{item:Closure}, \eqref{item:Dom3Disjoint}). 

\begin{lem}\label{lem:ImportedLem}
Let $(X,d)$ be a metric space  and let $P=(P_k)_{k \in K}$ be a tuple of  nonempty subsets in $X$. 
\begin{enumerate}[(a)]
\item\label{item:Monotone} $\Dom$ is antimonotone, i.e., $\Dom(R)\subseteq \Dom(S)$ whenever $S\subseteq R$; $\Dom^2$ is monotone, that is, $R\subseteq S\Rightarrow \Dom^2(R)\subseteq \Dom^2(S)$. 
\item \label{item:Closure} $\Dom(\overline{R})=\Dom(R)$.
\item\label{item:Dom3Disjoint} Suppose that $(X,d)$ is a geodesic metric space and that 
\begin{equation}\label{eq:r_k}
r_k:=\inf\{d(P_k,P_j): j\neq k\}>0 \quad \forall k\in K. 
\end{equation}
Then $(r_k/8)+(r_j/8)\leq d((\Dom^{\gamma}P)_k,(\Dom^{\gamma} P)_j)$ for any $j,k\in K$,$j\neq k$ and any
$\gamma\geq 2$. 
\end{enumerate}
\end{lem}

\begin{lem}\label{lem:disjoint}
Let $(X,d)$ be a metric space, let $P=(P_k)_{k\in K}$ be a tuple of nonempty 
subsets of $X$. Suppose that $R=(R_k)_{k\in K}$ satisfies $P_k\subseteq R_k\subseteq X$ 
for each $k\in K$. 
 \begin{enumerate}[(a)] 
 \item\label{item:P_kR_kDisjoint} Suppose that $R\subseteq \Dom(R)$. If 
 $\overline{P_k}\bigcap \overline{P_j}=\emptyset$ whenever $j\neq k$, then
  $R_k\bigcap R_j=\emptyset$ for each $j,k\in K,\,j\neq k$.
\item\label{item:RDomR} 
Suppose that \eqref{eq:r_k} holds.
If $R\subseteq \Dom(R)$, then the components of $R$ satisfy 
$\max\{r_k,r_j\}/3\leq d(R_k,R_j)$ for each $j,k\in K, k\neq j$.
\item\label{item:RDom2R} Suppose that $R\subseteq \Dom^2(R)$, that $(X,d)$ is a geodesic metric space, 
and that \eqref{eq:r_k} holds. 
Then the components of $R$ satisfy $(r_k/8)+(r_j/8)\leq d(R_k,R_j)$ for each $j,k\in K, k\neq j$.
\end{enumerate}
\end{lem}

\begin{proof}
\begin{enumerate}[(a)]
\item Suppose by way of contradiction that $x\in R_k\bigcap R_j$ for some $j,k\in K$, $j\neq k$. 
Since $x\in R_k\subseteq (\Dom R)_k$ we have  
$d(x,P_k)\leq d(x,\bigcup_{i\neq k}R_i)\leq d(x,R_j)=0$, so $x\in \overline{P_k}$. In the same 
way $x\in \overline{P_j}$, a contradiction.

\item Let $j,k\in K$, $j\neq k$ and 
$x\in R_k\subseteq\dom(P_k,\bigcup_{i\neq k}R_i)$, $y\in R_j\subseteq\dom(P_j,\bigcup_{i\neq j}R_i)$.
This implies that $d(x,P_k)\leq d(x,R_j)\leq d(x,y)$ and $d(y,P_j)\leq d(x,y)$. Therefore  
\begin{equation*}
r_k\leq d(P_k,P_j)\leq d(P_k,x)+d(x,y)+d(y,P_j)\leq 3d(x,y).
\end{equation*}
 Thus $r_k/3\leq d(R_k,R_j)$. Similarly, $r_j/3\leq d(R_k,R_j)$.
\item From the monotonicity of $\Dom^2$ (Lemma \ref{lem:ImportedLem}\eqref{item:Monotone}) we have  $R\subseteq\Dom^2(R)\subseteq\Dom^4(R)$.
This,  Lemma \ref{lem:ImportedLem} parts \eqref{item:Monotone}-\eqref{item:Closure},  
the inclusion $P\subseteq R\subseteq (X)_{k\in K}$, and $\overline{P}=\Dom(X)_{k\in K}$ imply that 
$R \subseteq \Dom^4(X)_{k\in K}=\Dom^3(\overline{P})=\Dom^3(P)$. 
From Lemma \ref{lem:ImportedLem}\eqref{item:Dom3Disjoint} we conclude that 
\begin{equation*}
d(R_k,R_j)\geq d((\Dom^{3}P)_k,(\Dom^{3} P)_j)\geq (r_k/8)+(r_j/8)
\end{equation*}
for each $j,k\in K, k\neq j$.
\end{enumerate}
\end{proof}

\begin{lem}\label{lem:NeutralGeodesic}
Let $B=(B_k)_{k\in K}$ be a tuple of nonempty subsets in a geodesic metric space $(X,d)$ and 
suppose that 
\begin{equation}\label{eq:rho_kB_k}
\rho_k:=\inf\{d(B_k,B_j): j\in K,\,j\neq k\}>0 \quad \forall k\in K. 
\end{equation}
Then $N:=X\backslash\bigcup_{k\in K}B_k\neq \emptyset$. 
Moreover, $\bigcup_{k\in K} S_k\subseteq N$ where 
\begin{equation}\label{eq:N_k}
S_k=\{x\in X: d(x,B_k)< \rho_k,\,\, x\notin B_k\}. 
\end{equation}
\end{lem}
\begin{proof}
Let $j,k\in K$, $j\neq k$ and let $x\in B_k$, $y\in B_j$. Since $X$ is a geodesic metric space there exists 
an isometry $\gamma:[0,d(x,y)]\to X$ satisfying $\gamma(0)=x$ and $\gamma(d(x,y))=y$. 
Let $E$ be the inverse image of the part of the segment $[x,y]$ which does not meet $\overline{B_k}$ anymore, i.e.,
\begin{equation*}
E:=\{t\in [0,d(x,y)]:\,\, [\gamma(s),y]\cap \overline{B_k}=\emptyset \quad\forall s\in [t,d(x,y)]\}.
\end{equation*}
Since $y\in B_j$ and $y\notin \overline{B_k}$ (by \eqref{eq:rho_kB_k}) it follows that $d(x,y)\in E$. 
Thus $E\neq \emptyset$. Let $a=\inf E$. If $a=0$, then $\gamma(a)=x\in \overline{B_k}$. Otherwise $a>0$. 
Assume by way of contradiction that  
 $\gamma(a)\notin \overline{B_k}$. Since $\overline{B_k}$ is closed it follows that a small ball 
 around $\gamma(a)$ does not intersect $\overline{B_k}$. Because $\gamma$ is continuous, for any $t$ in a small segment 
 around $a$ the point $\gamma(t)$ is inside this ball and thus does not belong to $\overline{B_k}$. This contradicts the minimality of $a$.  Therefore $\gamma(a)\in \overline{B_k}$. 
  
Consider the line segment $(\gamma(a),y]$. Its length is at least $\rho_k$ by 
\eqref{eq:rho_kB_k} (since the distance between two sets is the distance between their closures). 
Since $\gamma$ is an isometry the length of $[a,d(x,y)]$ is at least $\rho_k$. 
Let $s\in(0,\rho_k)$ and let $z=\gamma(a+s)$. Then $z\in (\gamma(a),y]$ and 
$d(z,\overline{B_k})\leq d(z,\gamma(a))=s<\rho_k$. From the definition of $a$ there exists $b\in (a,a+s)\cap E$. 
Thus $[\gamma(b),y]\cap \overline{B_k}=\emptyset$ and in particular $z\notin B_k$. 
From \eqref{eq:rho_kB_k} (with $i$ instead of $k$) it follows that $z\notin \bigcup_{i\neq k}B_i$. 
Therefore $z\in N$  and in particular $N\neq\emptyset$. 

Finally, let $S_k$ be the shell defined in \eqref{eq:N_k} and let $x\in S_k$. From \eqref{eq:rho_kB_k} we see that 
$x\notin B_j$ for $j\neq k, j\in K$. In addition, $x\notin B_k$ by the definition of $S_k$. 
Hence $x\in N$ and $S_k\subseteq N$ for each $k\in K$. 
\end{proof}

\begin{thm}\label{thm:NeutralZone}
Let $(X,d)$ be a geodesic metric space and let $(P_k)_{k\in K}$ be a tuple of nonempty subsets of $X$. 
Assume that \eqref{eq:r_k} holds. Let $R=(R_k)_{k\in K}$ satisfy $P_k\subseteq R_k\subseteq X$ for 
each $k\in K$ and suppose that either $R\subseteq\Dom(R)$ or $R\subseteq \Dom^2(R)$. 
Then there exists a neutral region in $X$, i.e., $N:=X\backslash\bigcup_{k\in K}R_k\neq \emptyset$. 
In particular this is true when $R$ is a zone or a double zone diagram. Moreover, let 
\begin{equation*}
\beta_k=\left\{
\begin{array}{ll}
r_k/3, &\quad \textnormal{if}\,\, R\subseteq \Dom(R),\\
(r_k+\inf\{r_j:\, j\in K,\, j\neq k\})/8,& \quad\textnormal{if}\, R\subseteq \Dom^2(R)
\end{array}
\right.
\end{equation*}
for each $k\in K$. Then $\bigcup_{k\in K} S_k\subseteq N$, where for each $k\in K$, 
\begin{equation}\label{eq:N_kR_k}
S_k=\{x\in X: d(x,R_k)< \beta_k,\,\, x\notin R_k\}. 
\end{equation}
\end{thm}
\begin{proof}
This is a simple consequence of Lemma \ref{lem:NeutralGeodesic} with $B=R$ since \eqref{eq:rho_kB_k} 
is satisfied by Lemma \ref{lem:disjoint}\eqref{item:RDomR}-\eqref{item:RDom2R}. 
\end{proof}

\begin{figure}
\begin{minipage}[t]{0.45\textwidth}
\begin{center}
{\includegraphics[scale=0.52]{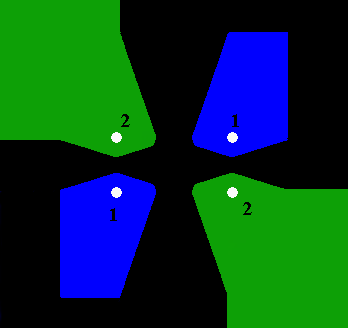}}
\end{center}
 \caption{The neutral region induced by a zone diagram of two sites, each 
 consists of 2 points, in a square in $(\R^2,\ell_1)$ (Example ~\ref{ex:ZDT}).}
\label{fig:ZD-4iterations-2sites-2inSite-0002-New}
\end{minipage}
\hfill
\begin{minipage}[t]{0.45\textwidth}
\begin{center}
{\includegraphics[scale=0.52]{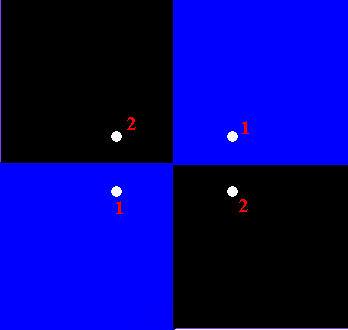}}
\end{center}
 \caption{The neutral region induced by a territory diagram $R$ of the setting of Example ~\ref{ex:ZDT}. 
 The second component of $R$ is $P_2$ and $R$ is not a double territory diagram.}
\label{fig:TD-2sites-2inSIte-0002}
\end{minipage}
\hfill
\begin{minipage}[t]{0.45\textwidth}
\begin{center}
{\includegraphics[scale=0.52]{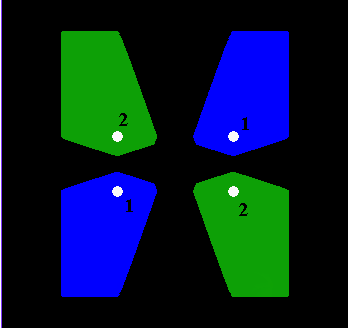}}
\end{center}
 \caption{The neutral region induced by the least double zone diagram $R$ 
 of the setting of Example~\ref{ex:ZDT}. $R$ is not a zone diagram.}
\label{fig:DZ-Min-0002-2sites2inSite-l1}
\end{minipage}
\hfill
\begin{minipage}[t]{0.45\textwidth}
\begin{center}
{\includegraphics[scale=0.52]{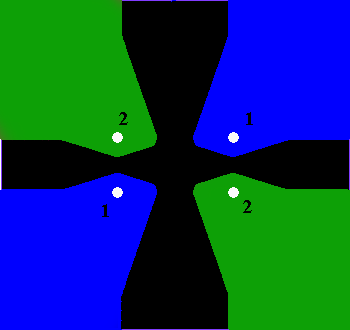}}
\end{center}
 \caption{The neutral region induced by the greatest double zone diagram $R$ of the setting of 
 Example~ \ref{ex:ZDT}. $R$ is a double territory diagram which is not a territory diagram.}
\label{fig:DZ-Max-0002-2sites2inSite-l1}
\end{minipage}
\end{figure}

\begin{expl}\label{ex:ZDT}
An illustration of Theorem \ref{thm:NeutralZone} is given in Figures 
\ref{fig:ZD-4iterations-2sites-2inSite-0002-New}-\ref{fig:DZ-Max-0002-2sites2inSite-l1} which 
also show some of the difference between the various notions. 
In all of these figures the setting is $X=[-6,6]^2$, $P_1=\{(2,1), (-2,-1)\}$,  $P_2=\{(-2,1), (2,-1)\}$, 
and the distance is the 2-dimensional $\ell_1$ distance. 
The (black) neutral region is clearly seen. Figures \ref{fig:ZD-4iterations-2sites-2inSite-0002-New}, 
\ref{fig:DZ-Min-0002-2sites2inSite-l1}, 
and \ref{fig:DZ-Max-0002-2sites2inSite-l1} were produced 
using the method described in \cite{ReemZoneCompute} and Figure \ref{fig:TD-2sites-2inSIte-0002} was produced directly. 
\end{expl}

\begin{expl}\label{ex:NeutralLineSegment}
From the proof of Lemma \ref{lem:NeutralGeodesic} and Theorem \ref{thm:NeutralZone} 
one obtains points in the neutral region by looking at certain parts of 
line segments connecting points located in different sites. 
This example shows that sometimes the neutral zone is nothing more than such a segment.  
In particular this example shows that the shells $N_k$  located 
around  the components of the (double) territory diagram (see \eqref{eq:N_kR_k}) 
can be very small. (Compare to the discussion in Section \ref{sec:ConcentrationOfMeasure}.)

Indeed, let $X_1=\{0\}\times (-2,3]$, $X_2=\{x\in \R^2: \|x-(0,-3)\|\leq 1\}$, 
and $X=X_1\cup X_2$, where $\|\cdot\|$ is the Euclidean norm. Define a metric $d$ 
on $X$ by $d(x,y)=\|x-y\|$ if $x$ and $y$ belong to the same component $X_i,\,i=1,2$, 
and $d(x,y)=\|x-(0,-2)\|+\|y-(0,-2)\|$ otherwise. Then $(X,d)$ is a geodesic metric 
space. Now let $P_1=\{(0,3)\}$, $P_2=\{(0,-3)\}$, $R_1=\{0\}\times [1,3]$, and 
$R_2=X_2\cup (\{0\}\times (-2,-1])$. Then $R=(R_1,R_2)$ is a zone diagram 
with respect to $P=(P_1,P_2)$ and the neutral region is $\{0\}\times (-1,1)$. 
See Figure \ref{fig:NeutralLineSegment}. 
\end{expl}
\begin{figure}
\begin{minipage}[t]{1\textwidth}
\begin{center}
{\includegraphics[scale=0.42]{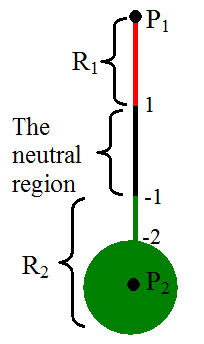}}
\end{center}
 \caption{The neutral region described in Example~ \ref{ex:NeutralLineSegment}. }
\label{fig:NeutralLineSegment}
\end{minipage}
\end{figure}

\begin{expl}\label{ex:NoNeutralZone}
Let $X=\{-1,0,1\}$ be a subset of $\R$ with the standard absolute value metric. 
Let $P_1=\{-1\}$, $P_2=\{1\}$. Let $R_1=P_1$, $R_2=\{0,1\}$. Then $R=(R_1,R_2)$ 
is a zone diagram (and hence also a territory diagram) but $R_1 \cup R_2=X$, 
violating Theorem~ \ref{thm:NeutralZone}. 
This is not surprising since $X$ is not a geodesic metric space. 
However, $R_1\cap R_2=\emptyset$, as predicted by Lemma \ref{lem:disjoint}\eqref{item:P_kR_kDisjoint}. 
This setting was mentioned in a different context in \cite[Example 2.3]{ReemReichZone}. 

In the same way, if $S_1=\{-1,0\}$ and $S_2=\{0,1\}$, then $S=(S_1,S_2)$ is 
a double zone diagram as a simple check shows (starting with observing that $\Dom(S)=(P_1,P_2)$). 
Now not only $S_1\cup S_2=X$, but also $S_1\cap S_2\neq \emptyset$. 
\end{expl}
\begin{expl}\label{ex:R=P}
Condition \eqref{eq:r_k} is necessary. Indeed, let $X=\R$ with the standard absolute 
value metric $d(x,y)=|x-y|$, let $K=X$, and let $P_k=k$, $k\in K$. Let $R=(P_k)_{k\in K}$. 
Then $(X,d)$ is a geodesic metric space, $R=\Dom(R)$, but $X\backslash(\cup_{k\in K}R_k)=\emptyset$. 
\end{expl}

\section{\bf Justifying the equilibrium interpretation of zone diagram}\label{sec:Interpretation}
One of the interpretations of zone diagrams, first suggested in \cite{AMTn} and then extended in 
\cite{ReemReichZone}, is a a certain equilibrium between mutually hostile kingdoms competing 
over territory. Kingdom number $k$ has a territory $R_k$ which has to be defended against attacks from the 
other kingdoms. Its site $P_k$ is interpreted as a castle, or, more generally, as a collection 
of army camps, castles, cities, and so forth. 
The sites remain unchanged and they are assumed to be located inside the kingdom and hence separated 
from each other. Due to various considerations (resources, field conditions, etc.), the defending 
army is located only in (part of) the corresponding site (unless the kingdom moves forces to attack another kingdom). 
 
 Assuming the time to move armed forces between two points is  proportional to the distance between 
 the points, it seems intuitively clear that if $R=(R_k)_{k\in K}$ is a zone diagram, then each point in each kingdom can be defended  at least as fast as it takes to attack it from any other kingdom, and no kingdom can enlarge its territory without  violating this condition. It also seems clear that the various territories are separated 
 by a no-man's land: the neutral territory. This was said explicitly in  \cite[p. 1183]{AMTn} where 
  the setting was the Euclidean plane and each site was a point. In \cite{ReemReichZone} the setting was 
 general and it was noted that counterexamples may exist in a discrete setting, but no further 
 investigation of the whole interpretation has been carried out. 
 
The goal of this section is to give a more rigorous justification to the above interpretation. It turns 
out that when the setting is similar to that of Theorem~ \ref{thm:NeutralZone}, then the interpretation holds. 

\begin{prop}\label{prop:Equilibrium}
Let $(X,d)$ be a geodesic metric space and let $P=(P_k)_{k\in K}$ be a tuple of nonempty subsets of $X$. 
Assume that \eqref{eq:r_k} holds. Suppose that $R=(R_k)_{k\in K}$ is a zone diagram corresponding 
to $P$. Then $R$ is an equilibrium in the above mentioned sense and there exists a neutral 
region separating its components.
\end{prop}
\begin{proof}
The existence of a neutral region was proved in Theorem \ref{thm:NeutralZone}. The proof actually 
shows that this region separates the regions $R_k,\,k\in K$ in the sense that any 
path connecting two points located in different components goes via the neutral region. 

As for the equilibrium interpretation, let $x$ be a point in some region $R_k$. 
By definition, $d(x,P_k)\leq d(x,\bigcup_{j\neq k} R_j)$. Since the time to move armed forces 
between any two points is  proportional to the distance between them, this shows that armed 
forces originating at $P_k$ will arrive to $x$ before any armed forces originating from 
another kingdom will arrive to $x$. This last fact is true in general, even in $m$-spaces \cite{ReemReichZone} 
(in which the distance function can be negative and does not necessarily satisfy the triangle inequality) 
and even if the sites are not mutually disjoint, although in this general case 
the interpretation looses something from its intuitiveness. 

It remains to prove that  no kingdom can enlarge its territory without violating the fast defense 
condition. More precisely, given any index $k\in K$ and any nonempty subset $A_k\subset X$ satisfying 
\begin{equation}\label{eq:A_k}
A_k\bigcap R_k=\emptyset=A_k\bigcap (\bigcup_{j\neq k} P_j), 
\end{equation}
if we let $\wt{R_k}=R_k\bigcup A_k$ 
and $\wt{R_j}=R_j\backslash A_k$ for any $j\neq k$, then there exist points in $\wt{R_k}$ which 
 cannot be defended fast enough by armed forces emanating from $P_k$: there is 
 some kingdom $R_j$, $j\neq k$ which can send its forces to these points and they will arrive 
 there before the defending forces from $P_k$ will arrive. In other words, it is not true that 
$\wt{R_k}\subseteq\dom(P_k,\bigcup_{j\neq k}\wt{R_j})$. 

To prove this, let $x\in A_k$ be arbitrary. 
Suppose for a contradiction that 
\begin{equation}\label{eq:x_in_domP_k}
d(x,P_k)\leq d(x,\bigcup_{j\neq k}\wt{R_j}). 
\end{equation}
First, by \eqref{eq:A_k} it follows that $x\notin R_k$. It must be that $x\notin R_j$ for any $j\neq k$. 
Indeed, assume by way of negation that $x\in R_j$ for some $j\neq k$. In particular $d(x,P_j)\leq d(x,R_k)$ 
and by Lemma \ref{lem:disjoint} we also know that $x\notin R_k$. Now observe the simple fact that 
the neighborhood $B(P_k,r_k/4)=\{y\in X: d(y,P_k)<r_k/4\}$ is contained in $R_k$ (a proof can be found 
in \cite{ReemZoneCompute} and a related claim also in \cite[Observation 2.2]{KMT}) and let $p\in P_k$ 
satisfy $d(x,p)<d(x,P_k)+(r_k/16)$. The segment $[p,x]$ starts at a point in $B(P_k,r_k/4)$ 
 and ends at a point outside this neighborhood and therefore the intermediate value theorem implies 
 that it intersects the boundary of  $B(P_k,r_k/4)$. The point of intersection $y$ is of 
 distance at least $r_k/4$ from $p$, otherwise it will be strictly inside  $B(P_k,r_k/4)$. 
The discussion above implies that 
\begin{equation*}
(r_k/16)+d(x,P_k)>d(x,p)=d(x,y)+d(y,p)\geq d(x,R_k)+(r_k/4)
\end{equation*} 
and hence, recalling that $D(x,R_k)\geq d(x,P_j)$, we have  
\begin{equation*}
d(x,P_k)>d(x,R_k)+(3r_k/16)\geq d(x,P_j)+(3r_k/16)>d(x,P_j). 
\end{equation*}
But this is impossible since 
we assumed that $d(x,P_k)\leq d(x,\bigcup_{i\neq k}\wt{R_i})$ and from \eqref{eq:A_k} we know 
that $P_j\subseteq \wt{R_j}\subseteq \bigcup_{i\neq k}\wt{R_i}$. This contradiction proves that 
$x\notin R_j$ for any $j\neq k$ and hence $A_k\cap (\bigcup_{j\neq k}R_j)=\emptyset$.

Finally $x$ cannot be in the (original) neutral region $N=X\backslash (\bigcup_{j\in K}R_j)$. Indeed, 
if $x$ is there then in particular $x\notin R_k=\dom(P_k,\bigcup_{j\neq k}R_j)$, i.e., $d(x,R_j)<d(x,P_k)$ 
for some $j\neq k$. But $R_j=\wt{R_j}$ since $A_k\cap R_j=\emptyset$ as proved above. Thus 
$d(x,\wt{R_j})<d(x,P_k)$, a contradiction to \eqref{eq:x_in_domP_k}. 
 Thus $x\notin R_k\bigcup(\bigcup_{j\neq k} R_j)\bigcup N=X$, an obvious contradiction. 
 Consequently \eqref{eq:x_in_domP_k} does not hold, i.e., $d(x,P_k)>d(x,\bigcup_{j\neq k}\wt{R_j})$ as claimed.
\end{proof}

\begin{remark}
When the space is no geodesic anymore a kingdom can enlarge its territory without violating the 
fast defense condition: just consider for instance  Example \ref{ex:NoNeutralZone} where $X=\{-1,0,1\}$ 
(or, if we allow attacks on the sites, even the more simple example 
where $X=\{-1,1\}$, $P_1=R_1=\{-1\}$,  $P_2=R_2=\{1\}$). Here it is worthwhile to kingdom 1 to try to 
capture the point $0$. However, one can argue against this example that the armed forces must jump out 
of the space  in order to arrive to the other kingdoms and if they do manage to do this, then they seem to 
appear there  ``out of the blue''. Hence it is implicitly assumed in the original interpretation that 
the space is ``continuous'', or,  in more precise terms, that it is a geodesic metric space or even a convex 
subset of a normed space. 
\end{remark}

\section{\bf A certain phenomenon related to measure concentration}\label{sec:ConcentrationOfMeasure}
We end this paper by showing that under simple conditions not only the neutral region is nonempty, but actually 
it can be quite large. Roughly speaking, given a double zone diagram of 
separated sites contained in the interior of a closed and convex world in (the Euclidean) $\R^m$, 
when the dimension of the space grows the volume of the neutral region 
becomes much larger than the volume of the ``interior regions'' (see the next paragraph). 
Hence, if the attention is restricted to these regions and the neutral one (as in conditional probability), 
then the neutral region occupies most of the volume. This property is related to the 
phenomena called ``concentration of measure'' \cite[pp. 165-166]{Gruber2007},\cite[pp. 329-341]{Matousek2002}. 
However, as can be seen here and there, the two phenomena are distinguished: 
for example, the discussion there is restricted to the Euclidean unit sphere with the normalized surface measure, 
the volume concentrates in a subset of a concrete form (near the equilateral, or, more generally, 
near the inverse image of the median of a Lipschtiz function), and one has a somewhat different bound on the ratio between 
the volumes of the various subsets (but there are some similarities, e.g., the dependence on the dimension is 
exponential in both cases). 

In what follows we are going to use the following terminology and notation. 
For a measurable subset $M\subseteq \R^m$ we denote by $\vol(M)$ the volume (Lebesgue measure) 
of $M$. Given a double zone diagram $R=(R_k)_{k\in K}$ induced by a tuple $(P_k)_{k\in K}$ of point sites 
located in the interior of the world $X$, a component $R_k$ is said to be an interior region if its distance from the boundary 
of $X$ is larger than some given positive parameter. Otherwise, it is said to be a boundary region.  
(It may be interesting to note that although we consider double zone diagrams, 
as a matter of fact, any such an object coincides  
with the unique zone diagram; this fact was not proved formally anywhere, but at least in our 
 specific setting it follows from the proof of a related theorem, namely \cite[Theorem 2.1]{KMT}.) 

\subsection{Imported results}
For the proofs in the sequel we will make use of the following three results which are special 
cases of results proved in \cite{ReemGeometricStabilityArxiv}.  
See, for instance, Theorem 8.2, Lemma 8.6, and Theorem 9.6 in v2 of the arXiv version of \cite{ReemGeometricStabilityArxiv} 
and also \cite[Theorem 3]{ReemISVD09}. See also Remark \ref{rem:Density} regarding \eqref{eq:BallRho}. 

\begin{thm}\label{thm:domInterval}
Let $P$ and $A$ be nonempty subsets of $X$. 
Then $\dom(P,A)$ is a union of line segments starting at 
  the points of $P$. More precisely, given $p\in P$ and a unit vector $\theta$, let 
\begin{equation}\label{eq:Tdef}
T(\theta,p)=\sup\{t\in [0,\infty): p+t\theta\in X\,\,\mathrm{and}\,\,
 d(p+t\theta,p)\leq d(p+t\theta,A)\}.
\end{equation}
Then
\begin{equation*}
\dom(P,A)=\bigcup_{p\in P}\bigcup_{|\theta|=1}[p,p+T(\theta,p)\theta].
\end{equation*}
In particular, if $P$ is composed of one point $p$, then we denote 
$T(\theta)=T(\theta,p)$ and we have 
\begin{equation*}\label{eq:dom_p}
\dom(P,A)=\bigcup_{|\theta|=1}[p,p+T(\theta)\theta].
\end{equation*}
.
\end{thm}

\begin{lem}\label{lem:BallRho}
Let $A$ be a nonempty subset of $X$. Let $p\in X$. Assume that  
\begin{equation}\label{eq:BallRho}
\exists \rho\in (0,\infty)\,\, \textnormal{such that}\,\,\forall x\in X\,\,\textnormal{the open ball}\,\, 
B(x,\rho)\,\,\textnormal{intersects}\,\,A.
\end{equation}
Then the mapping $T(\cdot)=T(\cdot,p)$ defined in Theorem \ref{thm:domInterval} satisfies 
$T(\theta)\in [0,\rho]$ for each unit vector $\theta$. 
\end{lem}

\begin{thm}\label{thm:contT}
Let $A$ be a subset of $X$. Let $p\in X$ be in the 
interior of $X$. Suppose that  $d(p,A)>0$. Suppose that \eqref{eq:BallRho} holds. 
Then the mapping $T(\cdot)=T(\cdot,p)$ defined in Theorem \ref{thm:domInterval} is continuous. 
\end{thm}

\subsection{The results}
The main result of this section is Theorem \ref{thm:COM}. 
Its proof is based on the following two lemmas. 

\begin{lem}\label{lem:SphericalTransformation}
Let $S^{m-1}$ be the unit sphere of $\R^m$, $m\geq 2$.  Let $f:S^{m-1}\to [0,\infty)$ be continuous. 
Let $V$ be the region defined by 
\begin{equation}\label{eq:V}
V=\{p+t\theta: \theta\in S^{m-1}, t\in[0,f(\theta)]\}.
\end{equation} 
Then $V$ is measurable and $\vol(V)=(1/m)\int_{S^{m-1}}({f(\theta)})^{m}d\theta$. 
\end{lem}
\begin{proof}
Let $F:[0,2\pi)\times [0,\pi]^{m-2}\to S^{m-1}$ be the spherical transformation mapping in a one-to-one 
way the rectangle $L_{m-1}=[0,2\pi)\times [0,\pi]^{m-2}$ onto the unit sphere. 
Let $G:[0,\infty)\times L_{m-1}\to \R^m$ be the  transformation of (translated) spherical coordinates 
defined by $G(r,\alpha)=p+rF(\alpha)$. More precisely,  $(x_1,\ldots,x_n)=G(r,\alpha_1,\ldots,\alpha_{n-1})$,
 where 
\begin{equation*}
\begin{array}{lll}
x_n&=&p_n+r\cos(\alpha_{n-1}), \\
x_{n-1}&=&p_{n-1}+r\sin(\alpha_{n-1})\cos(\alpha_{n-2}),\\
&\vdots & \\
x_3&=&p_3+r\sin(\alpha_{n-1})\ldots\sin(\alpha_3)\cos(\alpha_{2}),\\
x_1&=&p_1+r\sin(\alpha_{n-1})\sin(\alpha_{n-2})\ldots\sin(\alpha_2)\cos(\alpha_{1}), \\ 
x_2&=&p_2+r\sin(\alpha_{n-1})\sin(\alpha_{n-2})\ldots\sin(\alpha_2)\sin(\alpha_{1}),
\end{array} 
\end{equation*} 
and $p=(p_1,\ldots,p_n)$. The compactness of $S^{m-1}$ and the continuity of $f$ imply that $V$ is compact and 
hence measurable. We can write $V=G(W)$ where $W=\{(r,\alpha): r\in [0,f(F(\alpha))], \,\alpha\in L_{m-1}\}$. 
 The absolute value of the Jacobian of the smooth map $G$ is $|J|=r^{m-1}\Phi(\alpha)$ 
for some nonnegative and continuous function $\Phi:L_{m-1}\to \R$. This function is nothing but 
the change of variable factor (``Jacobian'') between the rectangle $L_{m-1}$ and $S^{m-1}$. In other words, 
it satisfies $d\theta=\Phi(\alpha)d\alpha$, namely $\int_{S^{m-1}} u(\theta)d\theta=\int_{L_{m-1}}u(F(\alpha))\Phi(\alpha)d\alpha$ for any continuous  function $u:S^{m-1}\to\R$ (this follows from the discussion on spherical coordinates in \cite[pp. 243-245]{Kuttler2009}). 
Thus, the change of variable formula and Fubini's theorem imply that 
\begin{multline*}
\vol(V)=\int_{V}dv=\int_{W}|J|(w)dw=\int_{L_{m-1}}\Phi(\alpha)\left(\int_0^{f(F(\alpha))}r^{m-1}dr\right)d\alpha\\
=\frac{1}{m}\int_{S^{m-1}}({f(\theta)})^{m}d\theta.
\end{multline*}
\end{proof}

\begin{lem}\label{lem:ConcentrationOfMeasureNeutral}
Let $X$ be a nonempty closed and convex subset of $\R^m$, $m\geq 2$, and let $(P_{k})_{k\in K}$ be 
 a tuple of point sites contained in the interior of $X$. 
 Suppose that $R=(R_{k})_{k\in K}$ is a double zone diagram 
corresponding to the sites. Let $N=X\backslash(\cup_{k\in K}R_{k})$ be the neutral 
region, the existence of which is guaranteed by Theorem ~\ref{thm:NeutralZone}. 
Suppose that for some $j\in K$ the region $R_j$ of the site $P_j$ is an interior 
region, namely there exists $\omega_j>0$ such that 
\begin{equation}\label{eq:omega_j}
\omega_j\leq d(R_j,\partial X).
\end{equation}
Assume that \eqref{eq:r_k} is satisfied (here $j$ and $k$ should be interchanged). Assume also that 
\eqref{eq:BallRho} holds with $A=\cup_{k\neq j}P_k$ and a positive number $\rho_j$.
  Then there exists a measurable subset $N_{j}\subseteq N$ satisfying  
\begin{equation}\label{eq:cm-1}
\vol(N_{j})\geq (c_j^m-1)\vol(R_{j})
\end{equation}
where $c_j=1+\min\{r_j/(32\rho_j),\omega_j/\rho_j\}$. Moreover, if for some  
$i\in K, i\neq j$ the region  $R_{i}$ is an interior region (with some parameter $\omega_i>0$) 
and \eqref{eq:r_k} and \eqref{eq:BallRho} hold for this $i$ (with $A=\cup_{k\neq i}P_k$ and $\rho=\rho_i$), 
then $N_{i}\cap N_{j}=\emptyset$ for the corresponding subset $N_{i}\subseteq N$. 
\end{lem}

\begin{proof}
Since $R=\Dom(\Dom(R))$ it follows that $R=\Dom(S)$ for some tuple $S$. In particular 
$R_j=\dom(P_j,A_j)$ for some subset $A_j$ of $X$. Thus, from Theorem ~\ref{thm:domInterval}, 
\begin{equation}\label{eq:R_j}
R_j=\bigcup_{|\theta|=1}[p_j,p_j+T_j(\theta)\theta] 
\end{equation}
where $T_j(\theta)=T_j(\theta,p_j)$ is defined in \eqref{eq:Tdef} and $P_j=\{p_j\}$. 
By Remark ~\ref{rem:TerritoryVoronoi} we know that $R_j$ is contained in the Voronoi region of $P_j$. 
Since Theorem \ref{thm:domInterval} holds for the Voronoi cell too, with, say, a corresponding 
function $\wt{T_j}$ defined in \eqref{eq:Tdef}, we have $T_j\leq \wt{T_j}$. Because \eqref{eq:BallRho} holds 
for this $\wt{T_j}$ (with $A=\cup_{k\neq j}P_k$), Lemma \ref{lem:BallRho} implies that $\wt{T_j}\leq \rho_j$ and thus 
$T_j(\theta)\leq \rho_j$ for any unit vector $\theta$.

Now fix a unit vector $\theta$. Let $a_{\theta}=p_j+T_j(\theta)\theta$, 
$b_{\theta}=a+\sigma_j\theta$ where $\sigma_j=\min\{r_j/32,\omega_j\}$. 
Consider the segment $(a_{\theta},b_{\theta}]$. We claim that 
any point in it belongs to the neutral region. Indeed, the definition of $T_j(\theta)$ 
 implies that the part of the ray in the direction of $\theta$ beyond $a$ is outside $R_j$. 
Any point $x\in (a_{\theta},b_{\theta}]$ is of distance at most $\omega_j$  from $R_j$ and hence 
\eqref{eq:omega_j} implies that $(a_{\theta},b_{\theta}]\subseteq X$. Finally, 
since $d(x,R_j)\leq r_j/8$ we conclude from Theorem \ref{thm:NeutralZone} and \eqref{eq:N_kR_k} 
that $(a_{\theta},b_{\theta}]\subseteq N$. 
The above discussion is true for any unit vector $\theta$. Hence the set 
\begin{equation}\label{eq:N_j}
N_j:=\bigcup_{\theta\in S^{m-1}}(a_{\theta},b_{\theta}] 
\end{equation}
is contained in $N$. Let $f:S^{m-1}\to (0,\infty)$ be defined by $f(\theta)=T_j(\theta)+\sigma_j$. 
Let $V$ be defined as in \eqref{eq:V}. Lemma \ref{lem:SphericalTransformation} and \eqref{eq:R_j} 
imply that $V$ and $R_j$ are measurable. Since  $N_j=V\backslash R_j$ it follows that it is measurable too. 
From Lemma \ref{lem:SphericalTransformation} we have 
\begin{equation}\label{eq:VolR_j}
\vol(R_j)=\frac{1}{m}\int_{S^{m-1}} (T_j(\theta))^m d\theta.  
\end{equation} 
Since $V$ is the disjoint union of $N_j$ and $R_j$, a second application of 
Lemma \ref{lem:SphericalTransformation} yields 
\begin{multline*}
\vol(N_j)+\vol(R_j)=\vol(V)=\frac{1}{m}\int_{S^{m-1}}{(f(\theta))}^{m}d\theta\\
=\frac{1}{m}\int_{S^{m-1}} (T_j(\theta)+\sigma_j)^m d\theta 
=\frac{1}{m}\int_{S^{m-1}} (T_j(\theta))^m \left(1+\frac{\sigma_j}{T_j(\theta)}\right)^md\theta\\
\geq \left(1+\frac{\sigma_j}{\rho_j}\right)^m \frac{1}{m}\int_{S^{m-1}} (T_j(\theta))^m  d\theta
\geq c_j^m \vol(R_j).
\end{multline*}
This implies \eqref{eq:cm-1}. Note that $T_j(\theta)>0$ for any $\theta$ since $p_j$ is in the interior 
of $X$ and any region of 
a double zone diagram of positively separated sites contains a small neighborhood around 
the site (see \cite[Observation 2.2]{KMT}). It remains to prove that $N_j\cap N_i=\emptyset$ whenever $i\neq j$. 
Indeed, let $x\in N_j\cap N_i$. The definition of $N_j$ (see \eqref{eq:N_j})
implies that $d(x,R_j)\leq \sigma_j\leq r_j/32$.  
In the same way $d(x,R_i)\leq r_i/32$. Thus $d(R_j,R_i)\leq (r_j/32)+(r_i/32)$, a contradiction with 
 Lemma~\ref{lem:disjoint}\eqref{item:RDom2R}.
\end{proof}

We are now able to prove Theorem \ref{thm:COM}. 
It says something which can be stated in a simple way intuitively (first paragraph) 
but, unfortunately, requires a somewhat long   formulation in order to achieve 
rigorousness (the remaining paragraphs). 

\begin{thm} \label{thm:COM}
Given a double zone diagram of separated point sites, if one restricts the attention 
to the part of the space occupied by the neutral region and the interior regions, 
then most of this volume concentrates at the neutral region as the dimension grows. 

More precisely, let $\rho>0$, $r>0$, and $\omega>0$ be given. 
Let $(X_m)_{m=2}^{\infty}$ be any sequence of closed and 
convex subsets satisfying $X_m\subset \R^m$. 
For each $m$ let $P_m=(P_{k,m})_{k\in K_m}$ be a tuple of sites in the interior of $X_m$.
Assume that for each $m$ 
\begin{equation}\label{eq:r}
\inf\{d(P_{k,m},P_{j,m}): j,k\in K_m\}\geq r.
\end{equation}
Assume that for each $m$ and for each $j\in K_m$ the relation 
\eqref{eq:BallRho} holds with the given $\rho$ and with $A_{j,m}=\bigcup_{k\neq j}P_{k,m}$. 
Let $R_m=(R_{k,m})_{k\in K_m}$ be a double zone diagram in $X_m$ 
corresponding to $P_m$. Let $J_m=\{j\in K_m: d(R_{j,m},\partial X_m)\geq\omega\}$. 
Assume that $J_m\neq \emptyset$. Let $F_m=\bigcup_{j\in J_m}R_{j,m}$ be the union of 
the interior regions (with parameter $\omega$). Let $N_m$ be the neutral region. Then 
\begin{equation}\label{eq:lim_1}
\lim_{m\to\infty}\frac{\vol(N_{m})}{\vol(F_m)+\vol(N_m)}=1
\end{equation}
with the agreement that $\infty/\infty=1$ if $\vol(N_m)=\infty$. 
As a matter of fact, if $c=1+\min\{r/(32 \rho),\omega/\rho\}$ and $\vol(F_m)<\infty$, then 
\begin{equation}\label{eq:Ocm}
\frac{\vol(F_{m})}{\vol(F_m)+\vol(N_m)}=O(c^{-m}). 
\end{equation}
\end{thm}

\begin{proof}
Since $N_m$ is the difference of two measurable sets it is measurable and hence $\vol(N_m)$ 
is well defined. Given $j\in J_m$, Lemma \ref{lem:ConcentrationOfMeasureNeutral} 
implies the existence of a subset $N_{j,m}\subseteq N_m$ whose volume satisfies \eqref{eq:cm-1} 
with $c_j=c=1+\min\{(r/(32\rho),\omega/\rho\}$. 
Since $N_{j,m}\cap N_{i,m}=\emptyset$ and $R_{j,m}\cap R_{i,m}=\emptyset$ 
whenever $i\in J_m, j\neq i$, it follows that 
\begin{equation}\label{eq:N_mF_m}
\vol(N_m)\geq \sum_{j\in J_m}\vol(N_{j,m})\geq (c^m-1) \sum_{j\in J_m}\vol(R_{j,m})=(c^m-1)\vol(F_m).
\end{equation}
If $\vol(N_m)<\infty$, then $\vol(F_m)<\infty$ by \eqref{eq:N_mF_m} and therefore 
\begin{equation}\label{eq:F_m}
\frac{\vol(F_m)}{\vol(F_m)+\vol(N_m)}\leq \frac{\vol(F_m)}{c^m\vol(F_m)}=c^{-m}.
\end{equation}
Thus \eqref{eq:Ocm} and \eqref{eq:lim_1} follow. Otherwise, \eqref{eq:lim_1} follows trivially by our agreement. 
Note that $\vol(F_m)\geq\vol(R_{j,m})>0$ whenever $j\in J_m$ because any region of 
a double zone diagram of positively separated sites contains a small neighborhood around 
the site (see \cite[Observation 2.2]{KMT}). Hence \eqref{eq:F_m} is well defined. 
\end{proof}
\begin{remark}\label{rem:ForbiddenZone}
In a very recent paper \cite[Section IV]{BKMK2012}, a similar bound between the volume of two 
sets related to the ones discussed above has been established independently. 
Here the discussed object was called ``the forbidden zone'' with region $R_i$ and a site $P_i$, namely the set 
\begin{equation}
F(R_i,P_i)=\{z\in X: d(z,y)<d(y,P_i)\,\,\textnormal{for some}\,\,y\in R_i\}
\end{equation}
and its volume was compared to the volume of $R_i$. 
The setting was a convex region $R_i$ in the Euclidean space $X=\R^m$ and $P_i$ was a point 
in $R_i$, but actually the same definition holds with respect to any given sets and in any metric space 
and in fact in the setting of Theorem ~\ref{thm:NoNeutral} and Remark ~\ref{rem:GeneralDistance}. 

Although one cannot use the results established 
in \cite{BKMK2012} since some of the inclusions mentioned there are not true when $X$ is a 
convex subset of $\R^m$ and not the whole space, one can still use the idea of multiplying the region 
(after translating the set so the site will be the origin) by a small enough positive constant. Similar bounds 
as established in Lemma \ref{lem:ConcentrationOfMeasureNeutral} can be obtained (the factor $1/32$ can be improved 
here and in Lemma \ref{lem:ConcentrationOfMeasureNeutral}). The advantage here is that there is 
no need to use some of the imported results such as Theorems \ref{thm:domInterval} and \ref{thm:contT}. 
In addition, one can avoid some (but not all) of 
 the proof of Lemma \ref{lem:ConcentrationOfMeasureNeutral} and all of 
 Lemma \ref{lem:SphericalTransformation}. However, it 
 seems that one cannot avoid  Lemma \ref{lem:BallRho} and  Theorem \ref{thm:COM}. On the other hand, the 
 advantage of the approach mentioned  in this paper is that explicit expression for the volume 
 of the region is given, namely  \eqref{eq:VolR_j}, and this 
 expression may be useful in other scenarios as well. 
\end{remark}

\begin{remark}\label{rem:Density}
In Theorem \ref{thm:COM} it was assumed that at least one interior region exists. 
Hence it is of some interest to formulate sufficient conditions on the sites and the world 
which will imply this existence. It turns out that one such a condition is simply that the 
underline subsets $X_m$ are not too thin and that the sites form a quite dense distribution in $X_m$. 

More precisely, suppose that for each $m$  there exists a point $x_m\in X_m$ satisfying 
$d(x_m,\partial X_m)\geq (8/3)\omega$. This holds, e.g., if $X_m$ is a cube or a ball 
having radius at least $(8/3)\omega$ and $x_m$ is the centre. 
Now suppose that the sites are distributed in $X_m$ is such a way that for any $x\in X_m$ the 
open ball $B(x,(2/3)\omega)$ meets at least one of the sites. In other words, for each $x\in X_m$ 
we have $d(x,A_m)<(2/3)\omega$ where  $A_m=\cup_{k\in K_m}P_{k,m}$. This is a similar condition to \eqref{eq:BallRho}. 
The above statements are our (not necessary optimal) sufficient condition.

Indeed, the above implies the existence of some $k\in K_m$ such that $d(x_m,P_{k,m})<(2/3)\omega$. 
We claim that the Voronoi cell  of $P_{k,m}$ has a distance at least $\omega$ 
from the boundary of $X_m$. Once this is proved one recalls that the corresponding region $R_{k,m}$ 
of a double zone diagram $R_m$ is contained in its Voronoi cell (Remark \ref{rem:TerritoryVoronoi}) 
and hence its distance from $\partial X_m$ is at least $\omega$. 

In order to prove the assertion about the Voronoi cell of $P_{k,m}$, let $x$ be any point in this cell. 
Assume for a contradiction that $d(x,\partial X_m)<\omega$. By the assumption on the distribution 
of the sites there exists a (point) site $P_{j,m}$ satisfying $d(x,P_{j,m})<(2/3)\omega$. Since 
\begin{equation*}
(8/3)\omega\leq d(x_m,\partial X_m)\leq d(x_m,P_{k,m})+d(P_{k,m},\partial X_m)<(2/3)\omega+d(P_{k,m},\partial X_m),
\end{equation*} 
it follows that $d(P_{k,m},\partial X_m)\geq 2\omega$.  But the definition of the Voronoi cell 
of $P_{k.m}$ implies that $d(x,P_{k,m})\leq d(x,P_{j,m})<(2/3)\omega$. Thus 
\begin{equation*}
d(P_{k,m},\partial X_m)\leq d(P_{k,m},x)+d(x,\partial X_m)<(2/3)\omega+\omega<2\omega,   
\end{equation*} 
a contradiction. 
\end{remark}

\begin{figure}
\begin{minipage}[t]{1\textwidth}
\begin{center}
{\includegraphics[scale=0.56]{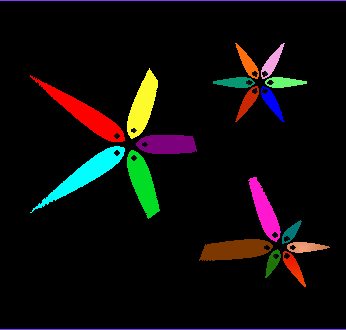}}
\end{center}
 \caption{The setting of Remark ~\ref{rem:InteriorRegions}: 17 sites in a rectangle in $(\R^2,\ell_2))$ and their corresponding 
 zone diagram. The regions are interior ones. The  bottom ``flower'' 
 has been obtained from the upper one by perturbing slightly the sites and then translating them as 
 a whole. }
\label{fig:ZD-Flowers-4Iteration-17Sites-0005}
\end{minipage}
\end{figure}

\begin{remark}\label{rem:InteriorRegions}
There are cases where all the regions are interior ones, as shown in Figure 
\ref{fig:ZD-Flowers-4Iteration-17Sites-0005} (related examples can be found in \cite[Fig. 4]{AMTn}). 
In such cases Theorem \ref{thm:COM} implies that the volume of the whole world concentrates at the 
neutral region as the dimension grows. It is interesting to find necessary and sufficient 
conditions which enforce this situation. Another interesting and related  phenomenon is 
a one shown in Figure \ref{fig:ZD-Flowers-4Iteration-17Sites-0005}: if a configuration of sites 
induces interior regions, then a small perturbation of the sites induce regions 
which are still interior ones. This property seems to be stable. 
However, experiments show that sometimes even slightly larger perturbations destroy this 
property. 

\end{remark}

\section{\bf Concluding remarks}\label{sec:ConcludingRemarks}
We end the paper with a few remarks about possible lines of investigation. 

Regarding the neutral Voronoi region, it may be interesting to discuss applications 
and variations of Theorem \ref{thm:NoNeutral} in the setting discussed there 
and also in other settings such as order-$k$ Voronoi diagrams  \cite[pp. 356-357]{Aurenhammer} 
or settings where there is a collection of distance functions corresponding 
to the site (instead of one global distance function) as in the cases 
of Voronoi diagrams induced by angular distances \cite{ATKT2006}, 
 weighted distances \cite[pp. 121-126]{OBSC} (additive, multiplicative), 
 power diagrams \cite{Aurenhammer1987}, \cite[pp. 380-386]{Aurenhammer}, and more. 

Regarding the neutral zone, it may be interesting to find better estimates for the size 
of this region than the ones given in Theorem ~\ref{thm:COM}. In particular, 
it is not clear whether the volume of the world concentrates 
at the neutral region if also boundary regions are taken into account. 
It is interesting to try to generalize Theorem ~\ref{thm:COM} for this case (or to find counterexamples) 
and also for the case of general sites. We believe that proving the existence of a neutral region 
in a context more general than Theorem ~\ref{thm:NeutralZone} is possible, with some caution 
(because of the counterexamples), e.g., for the case where several sites intersect, but we 
have no explicit result in this direction. It will also be interesting to 
answer the open problems mentioned in Remark ~\ref{rem:InteriorRegions}.

\section*{\bf Acknowledgments}
I thank Bahman Kalantari for sending me a copy of \cite{DeBiasiKalantaris2011},  
 Sedi  Bartz for helpful discussion related to Example \ref{ex:cluster}, and Iraj Kalantari and David 
 Menendez for helpful discussion related to Example \ref{ex:R=P}. 
\bibliographystyle{amsplain}
\bibliography{biblio}

\end{document}